\documentclass[draftcls,onecolumn]{IEEEtran}
\IEEEoverridecommandlockouts
\usepackage{graphics,graphicx,subfigure}
\usepackage{stfloats}
\usepackage{amsfonts,amssymb,amsmath,amsbsy,bm,paralist,theorem,ifthen,color}
\usepackage{calc}
\usepackage{array,multirow}
\usepackage{algorithmic,algorithm}
\usepackage{enumerate}
\usepackage[noadjust]{cite}
\usepackage{cases}
\newcommand\yb       {\ensuremath{{\bf x}}}
\newcommand{\taub}   {\ensuremath{{\bm\tau}}}
\newcommand\Tset     {\ensuremath{{\mathcal{T}}}}

\newcommand\Cset{\ensuremath{{{\mathbb C}}}}

\DeclareMathOperator*{\st}{s.t.}

\DeclareMathOperator*{\E}{\mathbb{E}}

\newtheorem{lemma}{Lemma}
\newtheorem{Proposition}{Proposition}

\newtheorem{Remark}{Remark}
\definecolor{orange}{RGB}{255,107,0}
\definecolor{green} {RGB}{0,  180,80}

\begin{document}
\title{A Low-Complexity Algorithm for Throughput Maximization in Wireless Powered Communication Networks}
\author{\IEEEauthorblockN{Qian Sun, Gang Zhu, Chao Shen, Xuan Li, and Zhangdui Zhong\\}
\IEEEauthorblockA{State Key Laboratory of Rail Traffic Control and Safety, Beijing Jiaotong University, Beijing, China 100044\\
Email: \{12120137, gzhu, shenchao, 12120099, zhdzhong\}@bjtu.edu.cn}
}
\maketitle

%In this paper we reconsider the time allocation for wireless powered communication networks. An optimal time allocation algorithm is proposed for the sum-throughput maximization problem based in the Jensen's inequality. Then we propose a low-complexity iteration algorithm for the min-throughput maximization problem, which promises a much better computation complexity than the state-of-the-art algorithm, thereby improving the power efficiency of WPCNs. Simulation results confirm the effectiveness of the proposed algorithm.

\begin{abstract}
This paper investigates a wireless powered communication network (WPCN) under the protocol of harvest-then-transmit, where a hybrid access point with constant power supply replenishes the passive user nodes by wireless power transfer in the downlink, then each user node transmit independent information to the hybrid AP in a time division multiple access (TDMA) scheme in the uplink.
The sum-throughput maximization and min-throughput maximization problems are considered in this paper. The optimal time allocation for the sum-throughput maximization is proposed based on the Jensen's inequality, which provides more insight into the design of WPCNs.
A low-complexity fixed-point iteration algorithm for the min-throughput maximization problem, which promises a much better computation complexity than the state-of-the-art algorithm. Simulation results confirm the effectiveness of the proposed algorithm.
\end{abstract}
\IEEEpeerreviewmaketitle

\section{Introduction}\vspace{-2pt}
Energy constrained wireless communication networks need some external charging mechanism to replenish their energy and remain active \cite{Medepally}. However, replacing or recharging batteries incurs a high cost and can be inconvenient or highly undesirable (e.g., for sensors embedded in building structures \cite{Ogawa} or inside the human body  \cite{ZhangLife}).
Harvesting energy from the radio-frequency (RF) signals serves as a safe and promising option, which can be realized by {\it wireless power transfer} (WPT) \cite{Ho}. A reasonable WPT efficiency can be achieved by the state-of-the-art energy harvesting circuits \cite{Vullers}. % and some research works have been carried out to further improve the conversion efficiency from
%RF to direct current (DC).
 Moreover, the WPT technology has actually found its application in the
wireless sensor networks (WSNs). The interested readers are referred to \cite{Mitcheson,Le} and the references therein.

It is  important to well design the {\it wireless powered communication networks} (WPCNs) such that the wireless nodes can be efficiently powered. To this end, the optimal design of WPCNs has drawn a lot of attention in recent years.
Since the power and information can be drawn from the radio signals at the same time, some research works were carried out on the hot topic of {\it simultaneous wireless information and power transfer} (SWIPT), e.g., \cite{Varshney,Grover,Ho,ShiQJ2013,Lidd2013,ChaoShen2013}. It was shown that there exist nontrivial tradeoffs between the WPT and {\it wireless information transmission} (WIT) efficiency under the single-input single-out (SISO) flat fading channel \cite{Varshney}, and the frequency-selective fading channel \cite{Grover}, as well as the multi-input multi-output flat fading channel \cite{Ho}.
The optimal beamforming designs for SWIPT were studied under the two-way relay channel \cite{Lidd2013} and the MISO interference channel \cite{ShiQJ2013,ChaoShen2013}. Besides, it is also  important to improve the WPCN system performance from the network perspective \cite{HuangKB2012,Lee2012,Shi,Ju}.

In this paper, we reconsider the scenario in \cite{Ju}, where a harvest-then-transmit protocol was  proposed for a point-to-multipoint network and two system utilities were taken into account, i.e., the sum-throughput maximization and max-min throughput.
Motivated by the fact that the circuit power dissipation is not negligible especially for the WPT based nodes \cite{CircuitPower}, it is hence very important to develop fast and simple enough algorithms for WPCNs.
The contribution of this paper is to propose a low-complexity algorithm for the optimal time allocation of the system in \cite{Ju}, such that less energy harvested by WPT will be required for the the baseband processing. This work also provides more insight into the optimal time allocation for the  system considered in \cite{Ju}.

This rest of the paper is organized as follows. The system model is described in Sec. II. Then in Sec. III we propose two fast optimal time allocation algorithms for the sum-throughput maximization and the min-throughput maximization, respectively. Sec. IV presents the simulation results for comparison with the algorithm in \cite{Ju}. Finally, Sec. V concludes the paper.

\section{System Model}%\vspace{-4pt}
This paper considers a SISO WPCN scenario \cite{Ju}, as shown in Fig. \ref{fig:1},  which consists of a single-antenna hybrid AP (HAP)  and $K$ single-antenna user nodes, denoted by U$_i$ for $i\!=\!1,\ldots,K$.
The network operates in a time division multiple access (TDMA) fashion. Assume that the frame duration is normalized to be 1. %slotted frame structure with normalized duration is shown in Fig. \ref{fig:2}. The
At the first $\tau_0\in[0,1]$ fraction of time, the HAP transfers the power to the $K$ user nodes in the downlink (DL).
%Then, $\tau_i\in[0,1]$ fraction of time is allocated to U$_i$, which transmits its independent information to the HAP with the energy harvested at the initial slot in the uplink (UL), for $i=1,\ldots,K$.
{Then in the uplink (UL), U$_i$ sends its independent information to the HAP within $\tau_i\in[0,1]$ fraction of time by the energy harvested at the initial slot, for $i=1,\ldots,K$.} Here, the total time constraint reads
\begin{align}
    \sum_{i=0}^K \tau_i\leq 1.\label{P01:a}
\end{align}

\begin{figure}[!t]
  \centering
  \includegraphics[width=.55\linewidth]{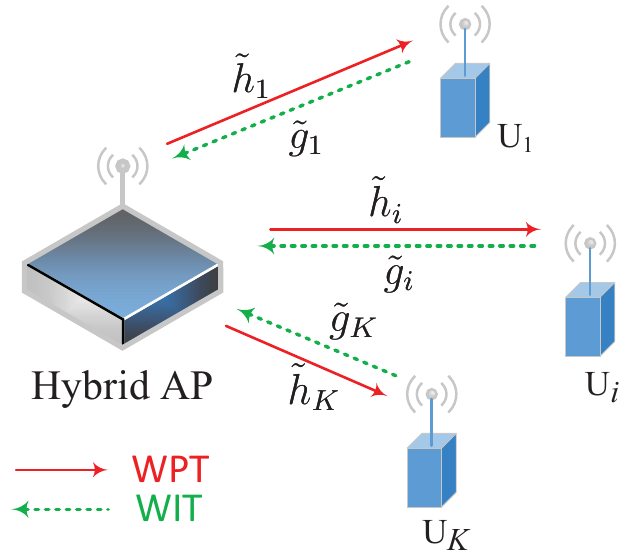}
  \caption{A wireless powered wireless network with wireless power transfer (WPT) in the downlink and wireless information transmissions (WITs) in the uplink.}\label{fig:1}\vspace{-10pt}
\end{figure}
{Assume that both DL and UL channels, denoted respectively by $\tilde h_i, \tilde g_i\in\Cset$, are quasi-static flat-fading. The HAP has perfect knowledge of all channel state information.}
During the DL WPT phase, the transmitted WPT signal from the HAP is denoted by $x_0\in\Cset$, which is subject to the average power limit, i.e., $\E [|x_0|^2]\leq P_{\max}$. Thus, the energy harvested at U$_i$ can be expressed by
\begin{align}
    E_i =\tau_0 \xi_i \E[|\tilde h_i x_0|^2]\leq \tau_0 \xi_i h_i P_{\max},\label{P03:a}
\end{align}
where $\xi_i\in(0,1)$ is the energy harvesting efficiency at U$_i$, and $h_i\triangleq |\tilde h_i|^2$, for  $i = 1,\cdots, K$.

In the subsequent UL WIT phase, all nodes transmit their independent information to the HAP with the energy harvested in the DL WPT phase, assuming that all the energy harvested is used.
Let the WIT signal transmitted by U$_i$ be $x_i\sim \mathcal{CN} (0,P_i)$. Then its average power is limited by
\begin{align}
    P_i = \E[|x_i|^2] = \frac{E_i}{\tau_i},~i = 1,\cdots, K,\label{P04:a}
\end{align}
and the received signal $ y_i$ at the HAP in the $i$th UL slot is % given by
   \begin{align}
     y_i =\tilde g_i x_i+z_i, i = 1,\cdots, K,\label{P05:a}
  \end{align}
where $z_i\sim \mathcal{CN}(0,\sigma_i^2)$ represents the additive Gaussian noise at the HAP. Therefore, the achievable throughput of $U_i$ in bits/second/Hz (bps/Hz) can be expressed as
\begin{align}\!
    R_i(\tau_0,\tau_i)=\tau_i\log_2\left(1+\frac{g_iP_i}{\Gamma\sigma_i^2}\right)
    \leq\tau_i\log_2\left(1+\gamma_i\frac{\tau_0}{\tau_i}\right),\label{P06:a}
\end{align}
where $g_i\triangleq|\tilde g_i|^2$, $\gamma_i =\frac{\xi_i h_ig_iP_{\max}}{\Gamma\sigma_i^2}$, for $i\! = \!1,\!\ldots,\! K$,  $\taub\!\triangleq\![\tau_0,\ldots,\tau_K]^T$ and $\Gamma$ denotes the signal-to-noise ratio (SNR) gap due to a practical modulation and coding scheme used.

\section{Problem Formulation}
\subsection{Time Allocation for Sum-Throughput Maximization}
In this subsection, let us focus on the optimal time allocation with the objective of sum-throughput maximization. From \eqref{P06:a}, the  problem can be formulated as
  \begin{subequations}\label{P1}
  \begin{align}
    \max_{{\bm\tau}}~&~\sum_{i = 1}^K R_i(\tau_0,\tau_i)\label{P1:a}\\[-3pt]
   \st~&R_i(\tau_0,\tau_i)\leq\tau_i\log_2\left(\!1+\gamma_i\frac{\tau_0}{\tau_i}\!\right)\!,\forall i=1,\ldots, K,\label{P1:a2}\\[-3pt]
        &\tau_i\geq 0, ~\forall i=0,\cdots, K, \label{P1:b}\\[-3pt]
        &\sum_{i = 0}^K \tau_i\leq 1.\label{P1:c}
  \end{align}
\end{subequations}
which admits a semi-closed form solution as stated below\vspace{-6pt}
\begin{lemma}{\rm\bf(\cite[{Proposition 3.1}]{Ju})}\label{Lemma1}%\vspace{-2pt}
  The problem \eqref{P1} is convex and the optimal time allocation is given by
    \begin{subnumcases}{\tau_i^*=}
        \frac{z^*-1}{A+z^*-1},& $i=0$,\\
        \frac{\gamma_i}{A+z^*-1},& $i=1,\cdots K$,
    \end{subnumcases}
    where $A =\sum_{i = 1}^K \gamma_i > 0$, and $z^*>1$ is the solution to $z\ln z- z + 1 = A$.
\end{lemma}\vspace{-4pt}

Notice that in Lemma \ref{Lemma1} the concavity of \eqref{P1} is proved by the negative semidefiniteness of the Hessian of $R_i(\taub)$, while the optimal time allocation is obtained from the Karush-Kuhn-Tucker (KKT) conditions of \eqref{P1}. However, we would like to remark that Lemma \ref{Lemma1} can indeed be established in a simple but more insightful way, as detailed below.

Firstly, it can be readily shown that the inequality constraints \eqref{P1:a2} and \eqref{P1:c} hold with equality at the optimal solution. So the function $R_i(\tau_0,\tau_i)\triangleq\tau_i\log_2\big(\!1+\gamma_i\frac{\tau_0}{\tau_i}\!\big)$ can be regarded as the perspective function of the concave function $\log_2(1+\gamma_i\tau_0)$, which demonstrates the concavity of $R_i(\taub)$ with respect to (w.r.t.) $\{\tau_0,\tau_i\}$.

Then, a semi-closed form solution for the optimal time allocation can be found by the Jensen's inequality. To show this, let us rewrite the problem \eqref{P1} as
  \begin{subequations}\label{P12}
  \begin{align}
    \max_{0\le\tau_0\leq 1}~f(\tau_0)\triangleq\max_{\tau_1,\ldots,\tau_K}&~\sum_{i = 1}^K \tau_i\log_2\left(\!1+\gamma_i\frac{\tau_0}{\tau_i}\!\right)\label{P12:a}\\[-3pt]
   \st~&\tau_i\geq 0, ~\forall i=1,\cdots, K, \\[-3pt]
        &\sum_{i = 1}^K \tau_i= 1-\tau_0.
  \end{align}
\end{subequations}

Due to the strict concavity of \eqref{P12:a} and thanks to the Jensen's inequality, for any given $\tau_0\in[0,1]$ the optimal $\{\tau_k^*\}_{k=1}^K$ is attained if and only if
   \begin{align}
    \gamma_i\frac{\tau_0}{\tau_i^*}={\rm SNR},~\forall i=1,\ldots,K,~{\rm and}~\sum_{i = 1}^K\tau_i^* = 1-\tau_0\label{P21:a}
  \end{align}
which yields
   \begin{align}
    {\rm SNR}=\frac{A\tau_0}{1-\tau_0},~\tau_i^*= \frac{\gamma_i}{A}(1-\tau_0).\label{P22:a}
  \end{align}
Substituting (\ref{P22:a}) into the sum throughput $f(\tau_0)$, $\tau_0$ can thus be chosen by solving the following problem
  \begin{align}
    \max_{{0\le\tau_0\leq 1}}~(1-\tau_0)\log\Big(1+ \frac{A\tau_0}{1-\tau_0}\Big),
  \end{align}
which is strictly concave and thus the unique maximizer can be fast found by golden section search.\vspace{-8pt}

\begin{Remark}
  The time fraction for each WIT slot is directly proportional to $\{\gamma_i\}_{i=1}^K$, and equal SNR is achieved at the optimal solution.
  The proposed approach, shedding some light on the optimal design of WPCN, can actually be further extended to the scenario for joint beamforming design and time allocation \cite{Shen2014}.
\end{Remark}

\subsection{Time Allocation for Min-Throughput Maximization}
In the previous subsection, the same SNR is achieved at the optimal solution. In view of the user fairness,
the following min-throughput maximization problem is considered\vspace{-2pt}
\begin{subequations}\label{P:SISO:maxmin}
  \begin{align}
    \max_{\{\tau_i\}_{i=0}^K}~&\min_{i=1,\ldots,K}~R_i\left(\tau_0,\tau_i\right)\label{P4:a}\\[-8pt]
    \st~&~0\leq\tau_i\leq 1,~\forall i=0,\ldots,K,~\sum_{i=0}^K\tau_i\leq 1.\label{P4:c}\vspace{-5pt}
  \end{align}
\end{subequations}

For convenience, we redefine $\taub\!\triangleq\![\tau_1,\ldots,\tau_K]^T$ and denote
\begin{align}\!\!\!\!\!
    \Tset(\tau_0) \!\triangleq \!\left\{\taub\left|\tau_i\in[0,1],\,\forall i=1,\ldots,K,\,{\bf 1}^T\taub\leq 1 - \tau_0\right.\right\}.\notag
\end{align}
Thus the problem \eqref{P:SISO:maxmin} can be alternatively expressed as
%\begin{subequations}\label{P::maxmin}
  \begin{align}\label{Problem14}
    \max_{\tau_0\in[0,1]}~g(\tau_0)\triangleq \max_{\taub\in\Tset(\tau_0)}&~\bar g(\tau_0,\taub), %\bar R\label{P::maxmin:a}\\
    %\st~&R_i\left(\tau_0,\tau_i\right)\geq \bar R,~\forall i.\label{P::maxmin:b}
  \end{align}
%\end{subequations}
where $\bar g(\tau_0,\taub)=\min\limits_{i=1,\ldots,K}~R_i\left(\tau_0,\tau_i\right)$.

It can be observed that $g(0)=g(1)=0$. However, what we are interested in is to show the following lemma:\vspace{-7pt}
\begin{Proposition}\label{Lemma::1}
  $g(\tau_0)$ is strictly concave w.r.t. $\tau_0\in[0,1]$.
\end{Proposition}\vspace{-4pt}
\begin{proof}
Please refer to the Appendix.
\end{proof}

Now let us focus on the time allocation for UL WITs with any given $\tau_0\in[0,1]$. By introducing a slack variable $\bar R$, the problem associated with $g(\tau_0)$ is equivalent to
\begin{subequations}\label{Problem19}
\begin{align}
   \max_{\taub,\,\bar R}~&\bar R\\[-3pt]
   \st~~&\tau_i\log_2\big(1+\gamma_i\tfrac{\tau_0}{\tau_i}\big)\geq \bar R,~\forall {i=1,\ldots,K},\label{Eq19b}\\[-2pt]
        & \tau_i\in[0,1],~\forall {i=1,\ldots,K},\\[-2pt]
        & {\bf 1}^T\taub\leq 1-\tau_0,\label{Eq19d}
\end{align}
\end{subequations}
which is convex and can thus be solved by, e.g., the subgradient approach in \cite{Ju}.
%which can be solved by bisection over $\bar R$ described in Algorithm 1.

But we are more interested in a fast fixed-point iteration algorithm in this paper. To this end, we notice that the throughput constraints in \eqref{Eq19b} and the time constraint \eqref{Eq19d} are active at the optimum, that can be easily verified by contradiction. Besides, the problem \eqref{Problem19} can be solved by bisection search over $\bar R$, and thereby checking the feasibility for the given $\bar R$, which is given by
\begin{subequations}\label{Problem20}
\begin{align}
   {\rm find} &~~~\taub\\
   \st  &~~ \tau_i\log_2\big(1+\gamma_i\tfrac{\tau_0}{\tau_i}\big)= \bar R,~\forall {i=1,\ldots,K},\label{Eq20b}\\
        &~~ \tau_i\in[0,1],~\forall {i=1,\ldots,K},~~ {\bf 1}^T\taub= 1-\tau_0.\label{Eq20d}
\end{align}
\end{subequations}
\vspace{-20pt}

%, otherwise, part of the time fraction of the node with maximal throughput can be equally allocated to the nodes with minimal throughput, such that the minimal throughput can be increased.
%
%It can be easily verified by contradiction that the inequalities in \eqref{Problem19} always hold with equality at the optimal solution. Otherwise, one can further maximize the minimal rate $\bar R$ due to the fact that the function $\tau_i\log_2(1+\gamma_i\frac{\tau_0}{\tau_i})$ is monotonically increasing in $\tau_i$.
%
\begin{lemma}\label{Lemma2}
For any given $\tau_0$ and $\bar R$, \eqref{Eq20b} gives rise to the unique solution for $\tau_i^*$, which can be obtained by the following fixed-point iteration with linear convergence rate:
\begin{align}\label{fixedpoint}
  \tau_i[n]=\frac{ \bar R}{\log_2\big(1+\gamma_i{\tau_0}/{\tau_i[n]}\big)},~\forall i=1,\ldots,K,
\end{align}
where $n$ is the iteration index.
\end{lemma}
\begin{proof}
  It can be proved that the function $f(x)=\frac{ \bar R}{\log_2\left(1+\gamma{\tau_0}/{x}\right)}$ is Lipschitz continuous with Lipschitz constant $L<1$, then $f(x)$ has precisely one fixed point. The detailed proof is omitted due to space limitations.
\end{proof}
%We note that for $x \in [0+\varepsilon,1-\varepsilon], f(x)\in [0+\varepsilon,1-\varepsilon]$, where $\varepsilon\rightarrow 0$. In addition, $f'(x)$ exist on $(0,1)$, with
%\begin{subequations}
%  \begin{align}
%    |f'(x)|&=\left|\frac{\gamma\tau_0\bar R}{x^2(1+\frac{\gamma\tau_0}{x})(\log_2(1+\frac{\gamma\tau_0}{x}))^2\ln 2}\right|\\
%    &=\left|\frac{\gamma\tau_0}{x(1+\frac{\gamma\tau_0}{x})\log_2(1+\frac{\gamma\tau_0}{x})\ln 2}\right|\\
%    &\leq L<1,
%  \end{align}
%\end{subequations}
%$f(x)$ converges to a unique fixed point $x$ in $[0+\varepsilon,1-\varepsilon]$.
%Since the $f'(x^*)\neq 0$, $f(x)$ converges linearly.

From Lemma \ref{Lemma2}, the feasibility problem \eqref{Problem20} is equivalent to check  whether $\taub^*\in\Tset(\tau_0)$ or not, where $\taub^*=[\tau_1^*,\ldots,\tau_K^*]^T$.
Furthermore, the optimal WPT time fraction $\tau_0$ can be obtained uniquely by golden section search based on the problem \eqref{Problem14}. The global optimum can be attained due to the strict concavity of $g(\tau_0)$ as stated in Lemma \ref{Lemma::1}.

To summarize, a fast time allocation algorithm for the problem \eqref{P:SISO:maxmin} is detailed in Algorithm \ref{Alg1}.

\begin{algorithm}[!b]
  \caption{A fast time allocation for the  problem \eqref{P:SISO:maxmin}}
  \label{Alg1}\small
\begin{algorithmic}[1]
\STATE Initialize $\tau_0,\ldots,\tau_K$, and $\epsilon>0$;
\REPEAT
\STATE Initialize $R_{\rm{min}} = 0$, $R_{\rm{max}} > \bar{R}^{*}$;
\REPEAT
\STATE $\bar R = \frac{1}{2} \left( R_{\min} + R_{\max} \right)$;
\STATE Compute $\{\tau_i^*\}_{i=1}^K$ by \eqref{fixedpoint};
\IF{$\sum_{i=1}^{K}{\tau_i^*}>1-\tau_0+\epsilon$}
\STATE ${R_{\max}} \leftarrow\bar R$;
\ELSIF{$\sum_{i=1}^{K}{\tau_i^*}<1-\tau_0+\epsilon$}
\STATE ${R_{\min}} \leftarrow\bar R$;
\ENDIF
\UNTIL{$\big|1-\tau_0-\sum_{i=1}^{K}{\tau_i^*}\big|\leq \epsilon$};
\STATE Update $\tau_0$ by golden section search;
\UNTIL {$\tau_0$ converges.}
\end{algorithmic}
\end{algorithm}

\section{Numerical Results And Discussion}
In this section, we will show the performance of the proposed algorithm for \eqref{P:SISO:maxmin} and compare it with the algorithm in \cite{Ju}.
The parameters are taken from \cite{Ju}, where the channel power gains are modeled as $h_i = g_i = 10^{-3}\rho^2_id^{-2}_i$ with
$\rho_i$ being the standard Rayleigh fading and $d_i$ being the distance between the HAP and U$_i$, while $\sigma_i^2=-100$ dBm and $\xi_i=0.5$, for $i = 1,\cdots, K$. Let the SNR gap $\Gamma = 9.8$ dB.
The same throughput accuracy $10^{-3}$ is adopted in the simulations.

Fig. \ref{fig:3} shows the average iteration time vs. the number of WIT nodes, where $d_i$ is uniformly distributed over $[5,20]$, $P_{\max}=5$ dBm and the results are averaged over 1000 random realizations. It turns out that %the average iteration time of
the proposed algorithm %increases
is much faster than the algorithm in \cite{Ju}, especially when the number of WIT nodes increases.

%\begin{figure}[htbp]
%\begin{minipage}[t]{0.3\linewidth}
%\centering
%\includegraphics[width=1.8in, height=1.65in]{fig3}
%\caption{Average iteration time vs. number of WIT nodes.}\label{fig:3}
%%\label{fig:side:a}
%\end{minipage}%
%\hspace{11ex}
%\begin{minipage}[t]{0.3\linewidth}
%\centering
%\includegraphics[width=1.8in, height=1.6in]{fig5}
%\caption{Individual throughput vs. the transmission power of the HAP.}\label{fig:5}
%%\label{}
%\end{minipage}
%\end{figure}

\begin{figure}[!t]\centering
  \includegraphics[width=.7\linewidth]{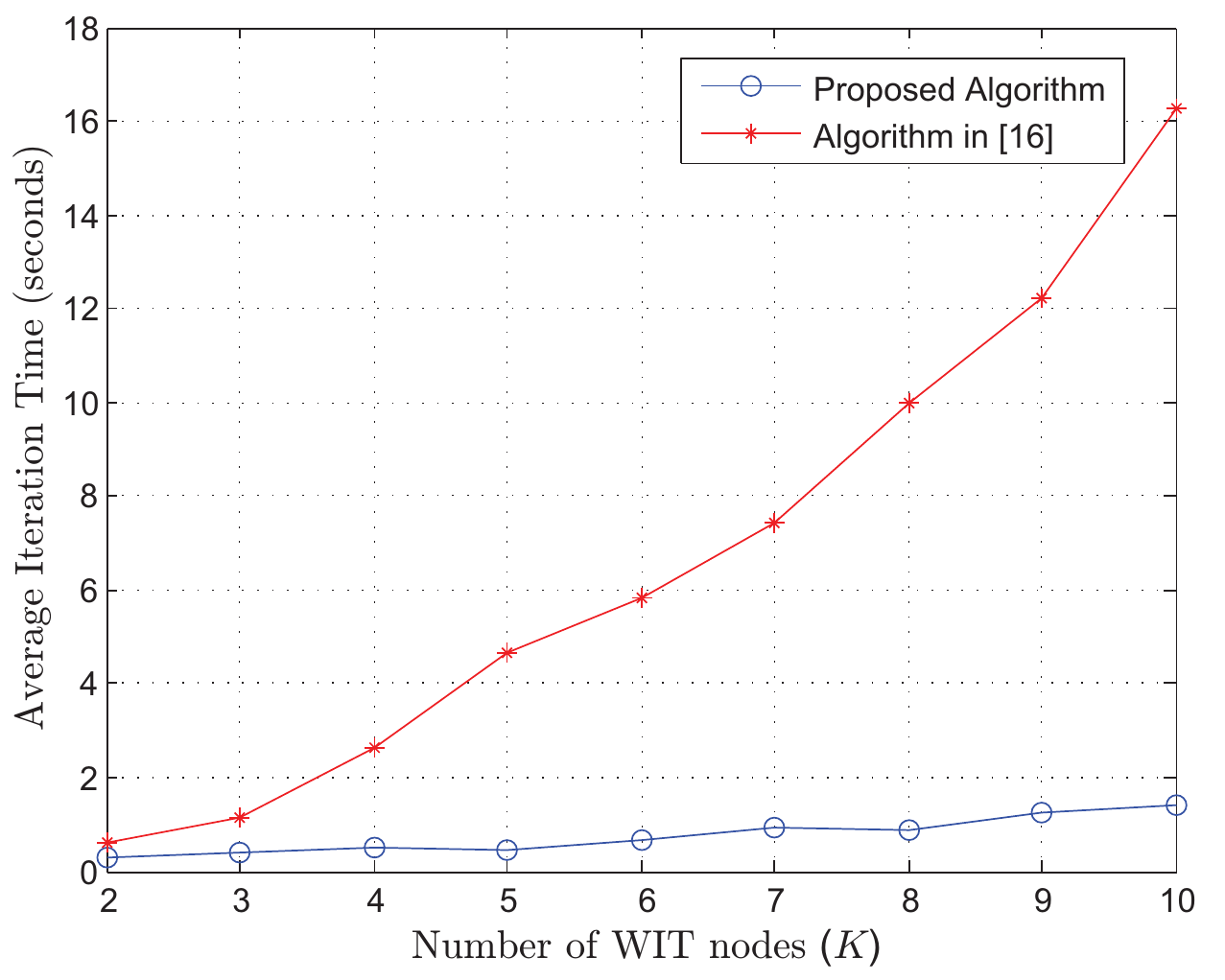}
 \caption{Average iteration time vs. number of WIT nodes.}\label{fig:3}
 \end{figure}
\begin{figure}[!t]\centering
  \includegraphics[width=.7\linewidth]{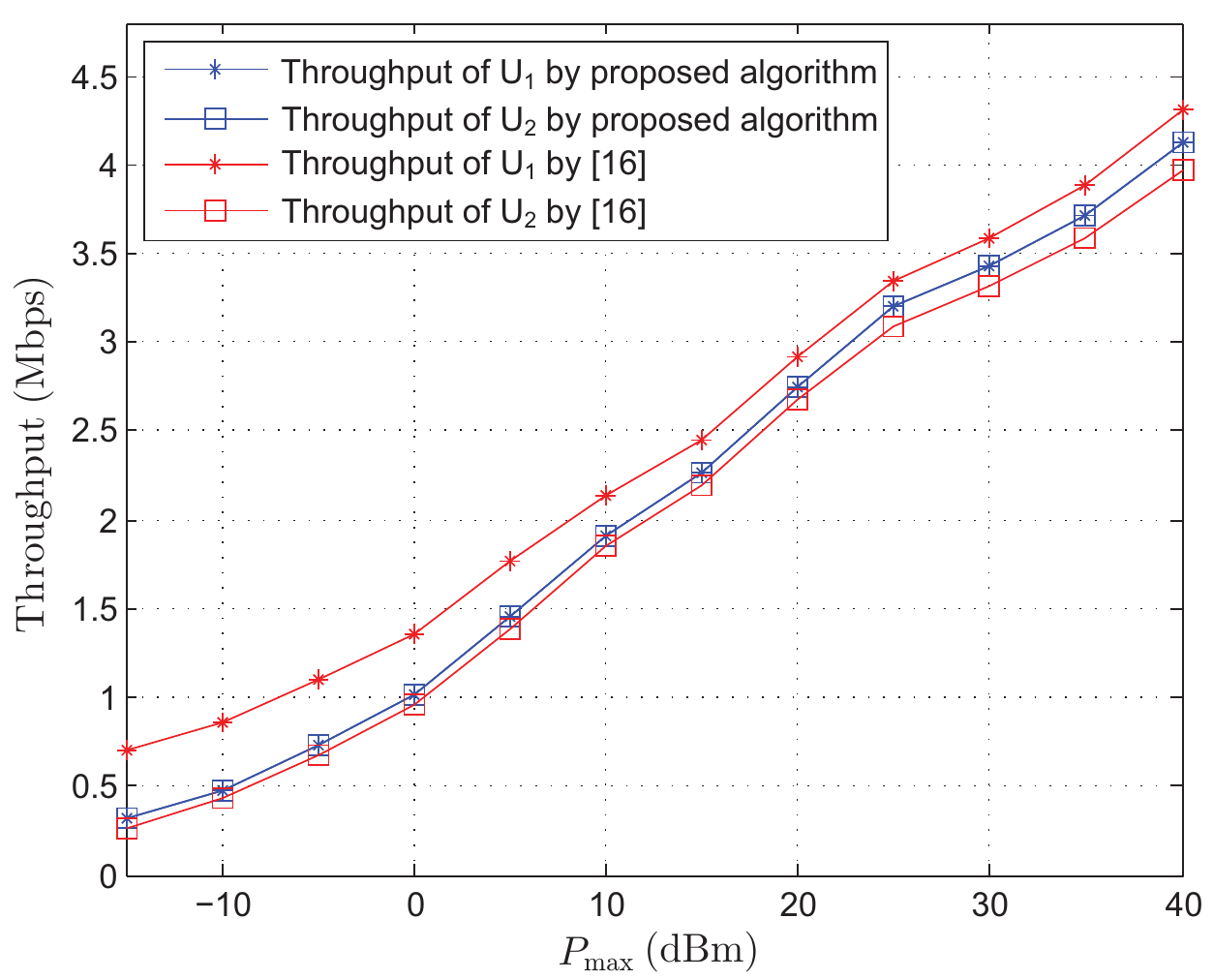}
  \caption{Individual throughput vs. the transmission power of the HAP.}\label{fig:5}
\end{figure}

In order to demonstrate the user fairness, in Fig. \ref{fig:5} we show the individual throughput vs. the power limit $P_{\max}$ with $K=2$ WIT nodes. The results are averaged over 1000  realizations with fixed $d_1=5$ m and $d_2=10$ m.
It can be seen that equal throughput is achieved by the proposed algorithm as expected. But there exists a throughput deviation between $R_1$ and $R_2$ by the algorithm in \cite{Ju}, which is contradicted with the previous analysis and may be caused by the iteration accuracy of the ellipsoid method in  \cite{Ju}.

\section{Conclusion}
In this paper we consider the time allocation for wireless powered communication networks. An optimal time allocation algorithm is proposed for the sum-throughput maximization problem based on the Jensen's inequality. Then we propose a low-complexity  algorithm for the min-throughput maximization problem, which promises a much better computation complexity than the state-of-the-art algorithm. Simulation results confirm the effectiveness of the proposed algorithm.

%sum-throughput maximization problem in \cite{Ju} by Jensen's inequality to make it more concision. Then, to overcome  the unfair rate allocation, minimum throughput (called common-throughput) maximization problem is proposed. As to the common throughput problem which can over come the unfair rate allocation in sum-throughput maximization problem, we propose an algorithm which converges linearly and shows better efficiency of both time and fairness to the algorithm in \cite{Ju}. Simulation confirms the effectiveness of the proposed algorithm. Furthermore, an example of some errors of the Lemma 4.1 is given to explain the disparity of time assignment of the proposed algorithm and the algorithm in \cite{Ju}.

\section*{Acknowledgment}
{\small
This work is supported by
the China Postdoctoral Science Foundation (No. 2013M530519),
the Key Projects of State Key Lab of Rail Traffic Control and Safety (No. RCS2012ZZ004),
the Natural Science Foundation of China (No. U1334202)
and the Key Grant Project of Chinese Ministry of Education (No. 313006).
Chao Shen is the corresponding author.}

\vspace{-2pt}\appendix
{\small
\begin{proof}[Proof of Lemma \ref{Lemma::1}]
Firstly, due to the fact that $R_i(\tau_0,\tau_i)$ is strictly concave w.r.t. $\{\tau_0,\tau_i\}$, the pointwise minimum function $\bar g(\tau_0,\taub)$ is also strictly concave w.r.t. $\{\tau_0,\taub\}$.

Secondly, we will prove the strict concavity of $g(\tau_0)$ over $[0,\,1]$. To show this, we need to prove that for any $t_1\in[0,1]$, $t_2\in [0,1]$, and $\theta\in (0,1)$, it admit\vspace{-4pt}
\begin{align}
  g(\theta t_1 + (1-\theta)t_2) > \theta g(t_1) + (1-\theta)g(t_2).
\end{align}
To this end, we first note that for any feasible points $t_1$ and $t_2$ ($t_1\neq t_2$), there exist $\yb_1\in\Tset(t_1)$ and $\yb_2\in\Tset(t_2)$ with\vspace{-4pt}
\begin{align}
  g(t_1)=\bar g(t_1,\yb_1), ~ g(t_2) = \bar g(t_2,\yb_2).
\end{align}
Note that $\yb_j\in\Tset(t_j)$ implies  $\yb_j\in [{\bf 0},{\bf 1}]$ and ${\bf 1}^T\yb_j\leq 1- t_j$ for $j=1,2$. Hence, for any $\theta\in[0,\,1]$, the point $\theta \yb_1 + (1-\theta)\yb_2$ satisfies\vspace{-4pt}
\begin{subequations}
  \begin{align}
    \theta \yb_1 + (1-\theta)\yb_2&\in \theta [{\bf 0},{\bf 1}] + (1-\theta)[{\bf 0},{\bf 1}] = [{\bf 0},{\bf 1}],\\[-2pt]\!\!\!\!\!
    {\bf 1}^T(\theta \yb_1 + (1-\theta)\yb_2) &= \theta {\bf 1}^T\yb_1 + (1-\theta){\bf 1}^T\yb_2 \\[-2pt]
    &\leq \theta (1-t_1) + (1-\theta)(1-t_2),
  \end{align}
\end{subequations}
which shows that the point $\theta \yb_1 + (1-\theta)\yb_2 \in\Tset\triangleq \Tset(\theta t_1 + (1-\theta)t_2)$.
Therefore, we have\vspace{-4pt}
\begin{subequations}
  \begin{align}\!\!\!\!\!
    g(\theta t_1 + (1-\theta)t_2)
    & = \max_{\yb\in\Tset}~\bar g(\theta t_1 + (1-\theta)t_2,~\yb)\\[-2pt]
    &\geq \bar g\left(\theta t_1 + (1-\theta)t_2,\,\theta \yb_1 + (1-\theta)\yb_2\right)\\[-2pt]
    &> \theta \bar g\left( t_1,\yb_1\right) + (1-\theta)\bar g\left( t_2,\yb_2\right)\\[-2pt]
    &= \theta g(t_1) + (1-\theta)g(t_2),
  \end{align}
\end{subequations}
where the second inequality is due to the strict concavity of $\bar g(\tau_0,\taub)$. This establishes the strict concavity of $g(\tau_0)$.
\end{proof}}

\bibliographystyle{IEEEtran}

\end{document}